\documentclass[a4paper,final,12pt]{article}
\usepackage[table]{xcolor}
\usepackage{enumitem,amsmath,amsfonts,amsthm,parskip,graphicx,url,amssymb,latexsym,bbm,xr,tikz,abstract,parskip,float,fontenc}
\usepackage{tikz, pgfplots, pst-plot}
\usepackage[justification=centering]{caption}
\usepackage{subcaption}
\usepackage{mathtools}
 \usepackage{mathabx}
\usepackage{geometry}
\usepackage[onehalfspacing]{setspace}
\usepackage{hyperref}
\hypersetup{colorlinks=true,
urlcolor=blue,
linkcolor=blue,
citecolor=blue}
\usepackage[section]{placeins}
\usepackage[english]{babel}
\usepackage[square,sort,comma,numbers]{natbib}
\usepackage[bottom,splitrule]{footmisc}

\theoremstyle{theorem}\newtheorem{Proposition}{Proposition}
\newtheorem{Definition}{Definition}
\theoremstyle{theorem}\newtheorem{Remark}{Remark}
\theoremstyle{theorem}
\theoremstyle{theorem}\newtheorem{Lemma}{Lemma}
\theoremstyle{theorem}\newtheorem{Theorem}{Theorem}
\theoremstyle{theorem}\newtheorem{Corollary}{Corollary}
\theoremstyle{definition}\newtheorem{example}{Example}

\newtheorem*{theoremthree}{Theorem \ref{th:theorem3}}
\newtheorem*{theoremfour}{Theorem \ref{th:theorem4}}
\newtheorem*{propone}{Proposition \ref{prop:proposition1}}
\newtheorem*{proptwo}{Proposition \ref{prop:proposition3}}

\makeatletter
\renewcommand\section{\@startsection {section}{1}{\z@}%
                                   {-1.5ex \@plus -1ex \@minus -.2ex}%
                                   {1.5ex \@plus.2ex}%
                                   {\normalfont\scshape}}
\renewcommand\subsection{\@startsection{subsection}{2}{\z@}%
                                     {-1.5ex\@plus -1ex \@minus -.2ex}%
                                     {0.25ex \@plus .2ex}%
                                     {\normalfont\itshape}}
\makeatother
\title{Final topology for preference spaces\footnote{I wish to thank Kim Border, Laura Doval, Federico Echenique, Mallesh Pai, Omer Tamuz, and participants of the LA theory fest for insightful discussions that helped shape this paper. All remaining errors are, of course, my own.}} %\\ %\footnotesize{(preliminary version, $\Pr(\exists\: \text{random typos})=1$.)} }     
\author{Pablo Schenone\footnote{Fordham University, Rose Hill Campus, 441 E. Fordham Road, Bronx, NY 10458} }
\begin{document}
\maketitle
\begin{abstract}

We say a model is continuous in utilities (resp., preferences) if small perturbations of utility functions (resp., preferences) generate small changes in the model's outputs. While similar, these two questions are different. They are only equivalent when the following two sets are isomorphic: the set of continuous mappings from \emph{preferences} to the model's outputs, and the set of continuous mappings from \emph{utilities} to the model's outputs. In this paper, we study the topology for preference spaces defined by such an isomorphism. This study is practically significant, as continuity analysis is predominantly conducted through utility functions, rather than the underlying preference space. Our findings enable researchers to infer continuity in utility as indicative of continuity in underlying preferences.
\end{abstract}
\normalsize
\newpage

\section{Introduction}\label{sec:intro}

\indent Economic models are mappings from exogenous variables to endogenous variables (henceforth, model outputs or outcomes). A common exogenous variable is the agent's \emph{preferences}, often modeled via a \emph{utility function}. Comparative statics exercises ask whether small perturbations to \emph{preferences} have a large impact on the model's outcomes. This exercise typically involves analyzing whether small perturbations to an agent's \emph{utility function} have a large impact on the model's outcomes.

\indent For example, in classical demand theory, market prices, consumer's wealth, and consumer's preferences are the exogenous variables, and the consumer's demand function is the model's outcome. Holding fixed all variables but the preferences, we want to know whether small perturbations to the consumer's \emph{preferences} will have a large impact on their demand. However, what we generally study is whether small perturbations to the consumer's \emph{utility function} will have a large impact on the consumer's demand. While similar, these two questions are different. They are only equivalent when the following two sets are isomorphic: the set of continuous mappings from \emph{preferences} to demands, and the set of continuous mappings from \emph{utilities} to demands. More generally, given any set of model outcomes, $Z$, continuity in preference and continuity in utilities coincide whenever the set of continuous functions from preferences to $Z$ is isomorphic to the set of continuous functions from utilities to $Z$.

\indent The isomorphism discussed above is a practical economic requirement for preference spaces. Preferences, being the economically meaningful object, are usefully represented by utility functions. If preference spaces admit more (or less) continuous mappings than utility function spaces, models may be continuous in preferences but not in utilities, or vice versa. Consequently, although utilities are order-faithful representations of preferences, they need not be a topologically-faithful representation of preferences. The aforementioned isomorphism guarantees that utilities are both an order and topologically faithful representation of preferences. Because economic theory often analyzes a model's continuity using the utility function space, ensuring utility functions are a topologically-faithful representation of preferences is important.

\indent Figure \ref{fig:final topology} depicts the isomorphism between the set of continuous function from \emph{preferences} into outcomes and the set of continuous functions from \emph{utilities} into outcomes. Let $F$ be the mapping that assigns to each utility function, $u$, the preference, $\succ$, that $u$ represents. Our objective is to find a topology for preference spaces that satisfies the following property. For each continuous function, $g$, that maps preferences into outcomes, a unique continuous function, $h$, exists that maps utilities into outcomes and makes the diagram in Figure \ref{fig:final topology} commute. That is, for all utility functions, $u$, $h(u)=(g\circ F)(u)$. Such a topology is known as the \emph{final topology}. 

\begin{figure}[H]
\begin{center}
\begin{tikzpicture}
%%nodes%%
\draw (0,4) node[left=3pt] {Space Of Preferences};
\draw (0,0) node[left=3pt] {Space Of Utility Functions};
\draw (4,4) node[right=3pt] { Outcome Space};
arrows
\draw[thick, ->, black] (-2,0.5) --node[left] {Representation Mapping, $F$} (-2,3.5);
\draw[thick, ->, black] (-1.75,0.5) -- node[below=3pt] {$h$} (3.5,3.5);
\draw[thick, ->, black] (0.5,4) -- node[above] {$g$} (3.5,4);
\end{tikzpicture}
\end{center}
\caption{The final topology on the preference space makes the diagram commute and preserves all  continuous mappings.\label{fig:final topology}}
\end{figure}
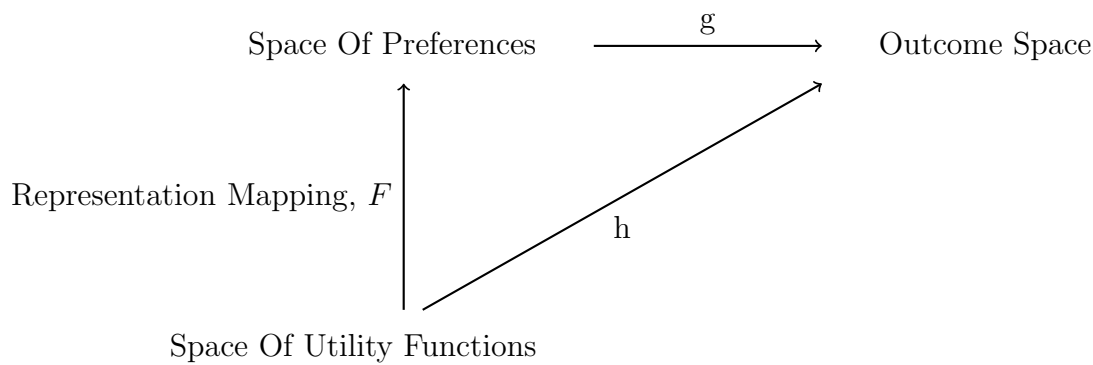

\indent In this paper, we characterize the final topology on preference spaces. Theorem \ref{th:theorem2} characterizes this topology in terms of its basis. Given any preference, $\succ$, and any finite subset of alternatives, $A$, that is strictly ranked according to $\succ$, all preferences that agree with $\succ$ over $A$ form an open set. For each fixed $\succ$, intersecting over a finite family of finite of such sets $A$ generates a typical basis element of the final topology.

\indent Theorem \ref{th:theorem2} allows us to study the properties of the final topology for preference spaces. We focus on two important properties: the Hausdorff property and path-connectedness. The Hausdorff property states that any two distinct preferences can be separated: that is, there is a neighborhood of the first and a neighborhood of the second that are disjoint. Path-connectedness states that given any two preferences, the first can be continuously transformed into the second. 

\indent The Hausdorff property and path-connectedness are essential for economic exercises, such as comparative statics or continuity analysis. First, to determine whether model outcomes converge as preferences converge to a limit, we need a well-defined notion of limit. The Hausdorff property guarantees that if a sequence of preferences admits a limit, the limit is unique. This uniqueness is essential for meaningful continuity analysis. For example, consider a sequence of consumers all of whom have the same preferences, $\succ$, and corresponding demands, $D$. If the Hausdorff property fails, this sequence of preferences could admit a limit $\hat{\succ}\neq\succ$, and the sequence of demands may admit a limit demand $\hat{D}\neq D$, none of which is economically meaningful. Second, if we want to ask how model outcomes change when we perturb preferences ``slightly", we need a notion of  a ``slight" change to a preference. Path-connectedness implies that any neighborhood of any preference contains at least another preference. Without this property, there would be (at least one) preference, $\succ$, such that the only ``close by" preferences would be $\succ$ itself. This would render any kind of comparative statics analysis meaningless. Moreover, path-connectedness means that given any two preferences, $\succ_0$ and $\succ_1$, we can continuously transform $\succ_0$ into $\succ_1$. Continuous transformation of one preference into another is the basis for homotopy-based comparative statics exercises (see Shiomura \cite{shiomura1998hicksian}, Borkovsky \cite{borkovsky2010user}, and Eaves and Schmedders \cite{eaves1999general} for applications).

\indent From Theorem \ref{th:theorem2}, Corollary \ref{corollary:properties} shows the final topology is path-connected but is not Hausdorff. The final topology is not Hausdorff because the space of all preferences admits preferences with indifferences. As such, the natural follow-up question is whether the subspace of strict preferences--i.e., preferences where no two alternatives are indifferent--is both Hausdorff and path-connected.

\indent Theorem \ref{th:theorem3} shows the space of strict preferences is Hausdorff, and Proposition \ref{prop:proposition1} shows the set of strict preferences is the largest set of preferences that is both Hausdorff and includes all strict preferences. However, the space of strict preferences is totally path-disconnected. That the space of strict preferences is totally path-disconnected means that each singleton is an open set, so each preference is topologically isolated. As such, comparative statics exercises based on small perturbations to preferences cannot be carried out. 

\indent Together, Theorems \ref{th:theorem2} and \ref{th:theorem3}, and Proposition \ref{prop:proposition1} show path-connectedness and the Hausdorff property are mutually exclusive properties for preference spaces, which either hold or fail, depending on whether we allow for indifferences. Therefore, assumptions on indifference curves have topological implications. As such, the topological properties of preference spaces are not a mere technicality, but carry substantial economic meaning and methodological restrictions.

\indent The analysis thus far does not impose any topology on the space of alternatives, $X$. Section \ref{sec:extensions} discusses how our results change if we impose a topology on $X$. As discussed in Section \ref{sec:extensions}, a topology on $X$ may be imposed in one of two ways. First, for each preference relation, $\succ$, we can endow $X$ with the weakest topology on $X$ that makes $\succ$ continuous (see Definition \ref{def:continuous preference}). Such a topology is generated by the upper and lower contour sets of $\succ$, and guarantees that optimal decisions are continuous. That is, if an alternative $x$ is preferred to an alternative $y$, then all alternatives sufficiently close to $x$ are also preferred to $y$. This generates a family of topological preference spaces, each of which differs only in the topology we impose for $X$. Second, we may endow $X$ with a preference-independent topology, $\mathcal{T}_X$, and restrict attention only to preferences that are continuous relative to $\mathcal{T}_X$. Proposition \ref{prop:proposition3} shows the first approach is not enough to guarantee that preference spaces are Hausdorff in the final topology. Instead, Theorem \ref{th:theorem4} shows the second approach is enough to restore the Hausdorff property if we restrict attention to the set of continuous and \emph{locally strict} preferences: preferences such that, for any pair of alternatives, $x$ and $y$, there is a close-by pair, $x'$ and $y'$, such that $x'$ and $y'$ are not indifferent (see Border and Segal \cite{border1994dynamic} and Definition \ref{def:locally strict} for a formal definition of locally strict preferences). Thus, imposing a preference-independent topology on $X$ ensures the Hausdorff property without limiting the analysis to strictly strict preferences. Comparing the final topology on preferences when we do and do not impose a topology on the space of alternatives highlights the role that topological assumptions on the space of alternatives play in shaping the topology of preferences.

\indent Summarizing, the paper presents two key findings. First, it highlights a trade-off between ensuring utility functions are topologically-faithful representations of preferences and maintaining intuitive topological properties of the preference space. Utilities should be a topologically-faithful representation of preferences, so the set of continuous functions from preferences to model outcomes should be isomorphic to the set of continuous functions from utilities to model outcomes. Yet, the topology resulting from such an isomorphism is either Hausdorff (if one restricts attention to strict preferences) or path-connected (if one allows for indifferences), but never both. Second, restoring the Hausdorff property for non-strict preferences necessitates imposing a preference-independent topology on the space of alternatives over which preferences are defined.

\paragraph{Related Literature}

\indent To the best of our knowledge, most of the work conducted on topologies for preference spaces comes from the literature on general equilibrium (see Debreu \cite{debreu1967neighboring}, Hildenbrand \cite{hildenbrand1970economies}, Grodal \cite{grodal1974note}, Mas-Colell \cite{mas1974continuous}). Debreu \cite{debreu1967neighboring} proposes a topology on the space of continuous preferences that is based on the Hausdorff semimetric; Hildenbrand \cite{hildenbrand1970economies} provides generalizations on Debreu's work. In these works, the space of alternatives, $X$, is endowed with an exogenous topology, and a preference $\succ$ is said to be continuous if each upper and lower contour set of $\succ$ is open in the topology of $X$.\footnote{By ``upper contour set" we mean the sets $Upper(x)=\{y: y\succ x\}$ for each $x\in X$; by ``lower contour set" we mean the sets $Lower(x)=\{y: x\succ y\}$ for each $x\in X$.} In Debreu's work, $X$ is assumed to be compact, so each continuous preference $\succ$ is a closed \emph{compact} subset of $X\times X$. By contrast, Hildenbrand assumes $X$ is locally compact, so $\succ$ is compact in the one-point Alexandroff compactification of $X$. The goal of these assumptions is to associate each preference with a closed subset of a compact space; therefore, one obtains a separable space by endowing the preference space with the Hausdorff semimetric. Because the space of preferences is given the topology of the Hausdorff semimetric, Hildenbrand calls it the \emph{topology of closed convergence}. Lastly, Kannai \cite{kannai1970continuity} takes the work of Debreu and analyzes the special case when preferences are continuous \emph{and monotone}.

\indent Outside of the realm of general equilibrium, Chambers, Echenique and Lambert \cite{chambers2019recovering} use tools from the above literature to address the following question. Suppose one observes a finite but large dataset generated by picking maximal elements out of a preference $\succ$; can one find a preference $\hat{\succ}$ such that $\hat{\succ}$ rationalizes observed choices and $\hat{\succ}$ is ``close" to the true preference $\succ$? They provide a positive answer under the topology of closed convergence when preferences are assumed to be \emph{locally strict}. Recently, Nishimura and Ok \cite{nishimura2023class} study topologies for preferences over a finite set of alternatives. The authors propose and axiomatically characterize a class of semimetrics which regard preferences as close (resp., far) if the choice behavior induced by these preferences is close (resp., far). 

\indent Our paper contributes to this literature in three main ways. 

First, in practice, continuity and comparative statics analyses are performed in the space of utility functions. Thus, understanding how the topology on the preference space interacts with the topology on the space of utility functions is important. Because most research on topologies for preference spaces consider only the preference space, the resulting topologies may not always ensure that continuity in preferences and continuity in utilities coincide. Instead, the defining feature of the final topology is that it universally guarantees that continuity in preferences and continuity in utilities coincide.
\color{black}

\indent Second, relative to the literature on general equilibrium, Theorem \ref{th:theorem2} shows the final topology is not the topology of closed convergence. The final topology differs from the topology of closed convergence for two reasons. First, whereas the general equilibrium literature is mainly concerned with continuity of Marshallian demand functions with respect to preferences, we are concerned with preserving continuity for any general outcome space, not just the space of Marshallian demands. Second, most of our exercise is carried out without imposing a topology on the space of alternatives, $X$, whereas the topology of closed convergence requires endowing $X$ with an exogenous topology.\footnote{For example, in the Marshallian demand case, $X=\mathbb{R}^N$ and is typically endowed with the Euclidean metric.}

\indent Third, our paper compares the topological properties obtained when imposing a topology on the space of alternatives, $X$, with those obtained when not imposing a topology on $X$. This highlights how the topological properties of $X$ shape the final topology on preference spaces. We show the most general way to obtain a final topology that is both Hausdorff and path-connected is to endow $X$ with an exogenous topology and focus attention only on continuous and locally strict preferences. Because the set of continuous and locally strict preferences varies with the topology the analyst chooses for $X$, the chosen topology for $X$ determines which preferences are represented by utility functions in a topologically-faithful manner.

\paragraph{Organization}
\indent The rest of this paper is organized as follows.
Section \ref{sec:model} presents the model and Section \ref{sec:results} our main results. Section \ref{sec:extensions} extends the model to allow for topologies on the set of alternatives. All proofs are in the appendix.

\section{Model}\label{sec:model}

\indent Let $X$ be the set of alternatives, $\mathcal{U}$ be the set of all utility functions on $X$, and $\mathcal{P}$ be the set of all binary relations on $X$ that admit a utility representation. That is, for each $\succ\in \mathcal{P}$, a function $u\in\mathcal{U}$ exists such that for all $x,y\in X$, $x\succ y$ if and only if $u(x)>u(y)$. Whereas most preference representation theorems depend on $X$ being a separable topological space, Dubra and Echenique \cite{dubra2000full} characterize conditions under which a preference admits a utility representation in the absence of topological assumptions on $X$. Let $F:\mathcal{U}\rightarrow\mathcal{P}$ denote the mapping that takes utility function $u$ to the preference $u$ represents; that is, for all $x,y\in X$, $u(x)>u(y)$ if and only if $x\:F(u)\:y$. We call the mapping $F$ the \emph{representation mapping}. Finally, let $\mathcal{T}_U$ denote the pointwise topology on $\mathcal{U}$.

\indent Our goal is to find a topology on $\mathcal{P}$ that preserves the structure of continuous mappings from $\mathcal{U}$ to any abstract set $Z$ and makes the diagram in Figure \ref{fig:final topology} commute. Under such a topology the continuous mappings from $\mathcal{P}$ to $Z$ are in one-to-one correspondence with the continuous mappings from $\mathcal{U}$ to $Z$. This requirement is captured by the definition below. 

\begin{Definition}\label{def:final topology}
A \emph{final topology} on $\mathcal{P}$ is a topology $\mathcal{T}_P$ that satisfies the following universal property: a function $g:\mathcal{P}\rightarrow Z$ is continuous if and only if a unique continuous function $h:\mathcal{U}\rightarrow Z$ exists such that the diagram in Figure \ref{fig:final topology2} commutes.
\end{Definition}

\indent Graphically, the final topology on $\mathcal{P}$ is the topology on $\mathcal{P}$ that makes the diagram in Figure \ref{fig:final topology2} commute and preserves the structure of continuous mappings from utilities to any economic variable of interest, $Z$. The uniqueness condition guarantees there is a bijection between the continuous functions $g$ and the continuous functions $h$; thus, our topology endows the preference space with exactly the same continuous functions as admitted by the utility space. 

\begin{figure}[H]
\begin{center}
\begin{tikzpicture}
%%nodes%%
\draw (0,4) node {$\mathcal{P}$};
\draw (0,0) node {$\mathcal{U}$};
\draw (4,4) node {$Z$};
arrows
\draw[thick, ->, black] (0,0.5) --node[left] {$F$} (0,3.5);
\draw[thick, ->, black] (0.25,0.5) -- node[below=3pt] {$h$} (3.5,3.5);
\draw[thick, ->, black] (0.5,4) -- node[above] {$g$} (3.5,4);
\end{tikzpicture}
\end{center}
\caption{$\mathcal{T}_P$ makes the diagram commute and preserves continuous mappings. \label{fig:final topology2}}
\end{figure}
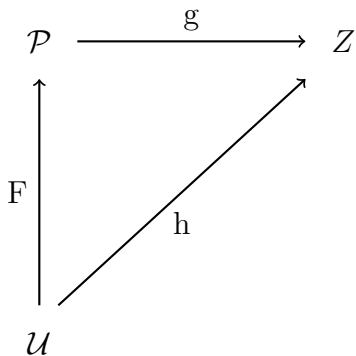 

\begin{Remark}
Moving forward, we need to distinguish between an instance of indifference \emph{within a preference} and the preference that corresponds to all alternatives being indifferent. We use the following notation:
\begin{itemize}
\item If $\succ\in \mathcal{P}$ we denote the weak part of $\succ$ by $\succsim$: $x\succsim y\Leftrightarrow y\nsucc x$.
\item If $\succ\in \mathcal{P}$ we denote the symmetric part of $\succ$ by $\sim$: $x\sim y\Leftrightarrow y\succsim x\:\text{ and }x\succsim y$.
\item Total indifference: if a preference satisfies that for all $x,y\in X$, $x\sim y$, then we use $\simeq$ to denote such preference.
\end{itemize}
\end{Remark}

\section{Results}\label{sec:results}

\indent In this section, we present our first two theorems. We briefly describe the theorems and follow up with formal statements. Proofs are provided in the Appendix. \\

\subsection{No restrictions on $\mathcal{P}$}\label{sec:allpreferences}

\indent We begin by introducing some notation that will help us construct the basis of $\mathcal{T}_P$, which we formally present in Theorem \ref{th:theorem2}. Let $\succ\in\mathcal{P}$ be any preference, and $A\subset X$ be any finite set such that for all $x,y\in A$ either $x\succ y$ or $y\succ x$, i.e., $x\nsim y$. We say that $A$ is \emph{$\succ$-strictly ranked}. Let the set $B(\succ,\:A)$ consist of all preferences that agree with $\succ$ on a set $A$ that is finite and $\succ$-strictly ranked. Formally, $B(\succ, A)=\{\hat{\succ}:\:(\forall x,y\in A)\: x\succ y\:\Leftrightarrow\: x\:\hat{\succ}\:y\}$. Sets of this form are open but do not constitute a basis of $\mathcal{T}_P$. A typical basis element is generated by fixing $\succ$ and taking intersection over finitely many such sets $A$. Formally, let $\alpha(\succ)$ be the collection of all finite families of such sets $A$, i.e., $\alpha(\succ)=\{\mathcal{A}\subset 2^X: \mathcal{A} \text{ is a finite family of finite and $\succ$-strictly ranked subsets of $X$}\}$. For any $\mathcal{A}\in\alpha(\succ)$ we denote by $B(\succ,\mathcal{A})$ the intersection of the sets $B(\succ, A)$ with $A\in\mathcal{A}$, i.e. $B(\succ,\mathcal{A})=\cap_{A\in\mathcal{A}}B(\succ, A)$. For each fixed preference, $\succ$, the set $B(\succ,\mathcal{A})$ is a typical basis element, as stated in Theorem \ref{th:theorem2} below.\\

\begin{Theorem}\label{th:theorem2}
Let $\mathcal{B}=\{B(\succ,\:\mathcal{A}):\:\succ\in\mathcal{P},\:\mathcal{A}\in\alpha(\succ)\}$. Then, $\mathcal{B}$ is a basis for $\mathcal{T}_{P}$.
\end{Theorem}

\indent The intuition behind Theorem \ref{th:theorem2} is simple. Lemma \ref{lemma:Fopen} in the Appendix shows $F$ is an open map. As such, we may characterize the basis elements of $\mathcal{T}_P$ by pushing forward the basis elements of $\mathcal{T}_U$ via $F$. Moreover, Example \ref{example:mathcalA} in the appendix shows why the sets $B(\succ,\:A)$ generate the basis of $\mathcal{T}_P$ although they do not form a basis for $\mathcal{T}_P$.

Armed with Theorem \ref{th:theorem2} we can study various properties of $\mathcal{T}_P$. Corollary \ref{corollary:properties} summarizes these properties:
\begin{Corollary}[Properties of the final topology]\label{corollary:properties}
The following hold: 
\begin{enumerate}
\item\label{itm:finite} The basis $\mathcal{B}$ of the final topology $\mathcal{T}_P$ can be generated by sets of the form $B(\succ, \{x,y\})$, where $x,y\in X$,
\item\label{itm:path} The topological space $(\mathcal{P},\mathcal{T}_P)$ is path-connected, and
\item\label{itm:hausdorff} The topological space $(\mathcal{P},\mathcal{T}_P)$ is not Hausdorff.
\end{enumerate}
\end{Corollary}

We now explain each of the observations in Corollary \ref{corollary:properties}.

\indent First, the basis $\mathcal{B}$ can be generated by sets of the form $B(\succ, \{x,y\})$, where $x,y\in X$. Starting from any preference, $\succ\in\mathcal{P}$, and any finite, $\succ$-strictly ranked set $A=\{x_1,...,x_N\}$, we obtain that $B(\succ,\{x_1,...,x_N\})=\cap_{i\neq j}B(\succ,\{x_i,x_j\})$. Therefore, finite intersections of sets of the form $B(\succ,\{x,y\})$ generate sets of the form $B(\succ,\:A)$, where $A$ is finite and $\succ$-strictly ranked. Taking another finite intersection over such sets $A$ generates each basis element in $\mathcal{B}$.

\indent Second, $\mathcal{P}$ is path-connected. Indeed, $\mathcal{U}$ is clearly path-connected.\footnote{Take any two functions $u,u'\in\mathcal{U}$. The transformation $t:[0,1]\times\mathcal{U}\rightarrow\mathcal{U}$ given by $t(s, u)=su+(1-s)u'$ continuously transforms $u$ into $u'$.} Since $\mathcal{P}=F(\mathcal{U})$ and $F$ is continuous, then $\mathcal{P}$ is path-connected as well.

\indent Finally, $\mathcal{T}_{P}$ is not Hausdorff. Theorem \ref{th:theorem2} implies that the only open neighborhood of $\simeq$ is $\mathcal{P}$, so $\simeq$ is topologically indistinguishable from any other preference. 

\subsection{Restricting $\mathcal{P}$ to strict preferences only}\label{sec:strictpreferences}

\indent In this section we examine preference spaces where the Hausdorff property holds. We concluded the previous section by showing the space $\mathcal{P}$ is not Hausdorff. We show $\mathcal{P}$ is not Hausdorff because any constant sequence admits $\simeq$ as a limit, prompting a natural question: is the space of strict preferences--that is, preferences where no two points are indifferent--Hausdorff? Furthermore, what is the largest set of preferences that maintains the Hausdorff property?

\indent Theorem \ref{th:theorem3} shows that in the subspace of strict preferences, the topology identified in Theorem \ref{th:theorem2} is Hausdorff but totally path-disconnected. To state Theorem \ref{th:theorem3}, we introduce some useful notation. Formally, let $\mathcal{P}^s=\{\succ\in\mathcal{P}:\: (\forall x,y\in X)\:x\nsim y\}$ and $\mathcal{U}^s=\{u\in\mathcal{U}:\:(\forall x,y\in X)\:u(x)\neq u(y)\}$. Analogously, we let $\mathcal{T}_{P^s}$ denote the final topology on $\mathcal{P}^s$ relative to $\mathcal{U}^s$. 

\begin{Theorem}\label{th:theorem3}
The space $(\mathcal{P}^s,\mathcal{T}_{P^s})$ is Hausdorff and totally path-disconnected.
\end{Theorem}

\indent Together, Theorems \ref{th:theorem2} and \ref{th:theorem3} present a fundamental trade-off. Whereas in some applications it is natural to consider non-strict preferences, this implies that any topology that universally preserves continuous mappings cannot be Hausdorff. Conversely, whereas the Hausdorff property is a natural property to ask of a topology, it comes at the cost of considering only strict preference spaces. A consequence of these observations is that topological conditions on preference spaces carry substantive behavioral assumptions about the indifference structures allowed in the preference space. 

\indent We now determine if the restriction from $\mathcal{P}$ to $\mathcal{P}^s$ is necessary or if a weaker condition on $\mathcal{P}$ still ensures that the Hausdorff property holds. Proposition \ref{prop:proposition1} shows the set of strict preferences is the largest set that is both Hausdorff and includes all strict preferences. Therefore, the restriction to $\mathcal{P}^s$ is the weakest restriction that preserves the Hausdorff property without excluding any strict preferences.

\begin{Proposition}\label{prop:proposition1}
Let $\mathcal{P}^0$ be be a preference space that is Hausdorff and satisfies that $\mathcal{P}^s\subset\mathcal{P}^0$. Then, $\mathcal{P}^0=\mathcal{P}^s$.
\end{Proposition}

\indent Because we are interested in providing a universal methodology under which utility functions are a topologically-faithful representation of preferences, the results so far do not presume any topology on the space of alternatives, $X$. In Section \ref{sec:extensions}, we show imposing a preference-independent topology on $X$ may resolve the tension between satisfying the Hausdorff property and path-connectedness.

\section{Imposing topologies on $X$}\label{sec:extensions}

\indent In this section, we study how topological assumptions on $X$ interact with the final topology on $\mathcal{P}$. In particular, we study whether the Hausdorff property can be recovered without restricting attention to strict preferences. 

\indent Whereas the results in Section \ref{sec:results} imply that making $\mathcal{P}$ Hausdorff requires limiting attention to strict preferences, endowing $X$ with a topology generates well-defined notion of \emph{locally strict} preferences (see Definition \ref{def:locally strict} below.) In this section we ask whether the Hausdorff property for preference spaces is guaranteed when restricting attention to locally strict preferences rather than strict preferences.  

\begin{Definition}\label{def:locally strict}
Let $(X,\mathcal{T}_X)$ be a topological space. Let $\succ$ be a preference defined on $(X,\mathcal{T}_X)$. We say $\succ$ is \emph{locally strict} if the following holds: for each $(x,y)\in X\times X$ such that $x\succsim y$ and for every neighborhood $V$ of $(x,y)$, there exists $(x',y')\in V$ such that $x'\succ y'$. \footnote{ See also Border and Segal \cite{border1994dynamic}.}
\end{Definition}

\begin{Remark}
Note the definition of locally strict preference is effective only when $x\sim y$. Otherwise, $x'=x$ and $y'=y$ always satisfy the definition.
\end{Remark}

\indent There are two natural ways in which we can impose a topology on $X$. First, we may assume $X$ does not have an exogenously given topology, and instead impose that, for each $\succ\in\mathcal{P}$, the topology on $X$ is the topology generated by the upper and lower contour sets of $\succ$, (see Definition \ref{def:upperlower}); we denote this topology by $\mathcal{T}_X(\succ)$. This approach generates a family of topological spaces, $(X,\mathcal{T}_X(\succ))_{\succ\in\mathcal{P}}$, where all preferences are continuous by definition and all topological properties of $X$ are driven by the decision problem under consideration. Second, we may assume $X$ comes endowed with an exogenously given topology; for example, if $X=\mathbb{R}^2$, we may endow $X$ with the topology induced by the Euclidean distance. Then, to guarantee that utility representations exist, we restrict attention to preferences that are continuous with respect to the given topology on $X$, i.e., preferences such that the upper and lower contour sets are open sets in the  topology of $X$. Restricting the space of preferences in this way is necessary because the topological properties of $X$ are independent of the decision problems being analyzed, so some preferences may not admit a utility representation. 

\indent We now analyze the final topology on $\mathcal{P}$ for both possibilities described above. 

\subsection{Endowing $X$ with the topology $\mathcal{T}_X(\succ)$ for each $\succ\in\mathcal{P}$}\label{sec:subjectivetopology}

\indent To formally define the topological space $(X,\mathcal{T}_X(\succ))$, we first define \emph{upper and lower contour sets}:\begin{Definition}\label{def:upperlower}
Let $X$ be a set, and let $\succ$ be a preference on $X$. We define the lower and upper contour sets as follows:
\begin{itemize}
\item[1.] $Lower (x)=\{y:x\succ y\}$,
\item[2.] $Upper(x)=\{y:y\succ x\}$.
\end{itemize}
\end{Definition}
Armed with this definition, we can define the topology $\mathcal{T}_X(\succ)$ as the topology on $X$ generated by the lower and upper contour sets of a give preference, $\succ$. That is, a set $G\subset X$ is open if and only if it can be written as a union of finite intersections of upper and/or lower contour sets.

\indent To motivate this section's choice of topology for $X$, consider the following simple example:

\begin{example}\label{example:pichun}
A decision maker chooses a bundle in $\mathbb{R}^2$ and considers the goods to be perfect substitutes. That is, they maximize the function $u(x_1,x_2)=x_1+x_2$. Suppose their wealth is $10$ and the prices of goods $1$ and $2$ are $p_1=1$ and $p_2=2$, respectively. Then, their optimal consumption is $(x_1,x_2)=(10,0)$. Suppose now the price of good 2 falls over time, following the pattern $p_2(t)=2-\frac{t}{t+1}$, with $t\in\mathbb{N}$. Because $p_2(t)>p_1(t)$ for all $t$, the optimal sequence of consumptions for this decision maker is $(x_1,x_2)(t)=(10,0)$ for all $t$. However, when $p_1=p_2=1$, $(x^*_1,x^*_2)=(0,10)$ is an optimal consumption bundle. The question is whether we should consider $(x^*_1,x^*_2)=(0,10)$ a valid limit point of the constant sequence $(10,0)_{t\in\mathbb{N}}$.
\end{example}

\indent One way to look at Example \ref{example:pichun} is to ask what the purpose of a topology is. One answer is that a topology on a space $X$ determines which functions involving $X$ are continuous. An economically meaningful requirement is that optimal decisions should be continuous. That is, if an alternative $x\in X$ is preferred to an alternative $y\in X$, then we should pick a topology on $X$ such that $x$ is still preferred to any $\hat{y}$ that is ``sufficiently close" to $y$, where ``sufficiently close" is determined by the decision maker's preferences. The topology $\mathcal{T}_X(\succ)$ is the smallest topology that guarantees continuity of a decision maker's choices. In the context of Example \ref{example:pichun}, this means that $(0,10)$ is indeed a valid limit to the constant sequence $(10,0)$: according to the decision maker's perception of $\mathbb{R}^2$, $(10,0)$ and $(0,10)$ are topologically indistinguishable because they share the same upper and lower contour sets, so any sequence that converges to the first converges to the second.

\indent Motivated by Example \ref{example:pichun}, rather than studying a fixed topology on $X$, this section studies the family of topological spaces, $(X,\mathcal{T}_X(\succ))_{\succ\in\mathcal{P}}$. 

\indent For the family of topological spaces, $(X,\mathcal{T}_X(\succ))_{\succ\in\mathcal{P}}$, locally strict preferences are characterized by the topological properties of their indifference classes. Specifically, for $\succ\in\mathcal{P}$ and $X$ endowed with the topology $\mathcal{T}_X(\succ)$, $\succ$ is locally strict in the space $(X, \mathcal{T}_X(\succ))$ if and only if no subset of any indifference class of $\succ$ is open. More formally, for any alternative $x\in X$, with $<x>$ denoting the indifference class of $x$, $\succ$ is locally strict if and only if no subset $G\subset<x>$ exists such that $G$ is open in $\mathcal{T}_X(\succ)$. The set of all such preferences is denoted by $\mathcal{P}^{ci}=\{\succ\in\mathcal{P}:\:(\forall x\in X), \: (\forall G\subset <x>),\:\text{ $G$ is not open}\}$. 

\indent Lemma \ref{lemma:locally strict} in the Appendix shows that every preference in $\mathcal{P}^{ci}$ is locally strict in the topological space $(X,\mathcal{T}_X(\succ))$. Moreover, if $\succ\notin \mathcal{P}^{ci}$ then $\succ$ is not locally strict in $(X,\mathcal{T}_X(\succ))$. If $\succ\notin \mathcal{P}^{ci}$, then $x\in X$ and $G\subset X$ exist such that $G\subset <x>$ and $G$ is open. For $(x,x)\in X\times X$, $G\times G$ is a neighborhood of $(x,x)$ where no pair $(z,z')\in G\times G$ satisfies $z\succ z'$, thus proving $\succ$ is not locally strict. In summary, a preference $\succ$ is locally strict relative to $(X,\mathcal{T}_X(\succ))$ if and only if $\succ\in\mathcal{P}^{ci}$.\footnote{For a more concrete example, consider the following: $X=\mathbb{R}$ and preferences $\succ$ are as follows: $1\sim 0\succ z$ for all $z\neq 0,1$. In this case $T_X(\succ)=\{\emptyset, X, \{0,1\}, X\setminus\{0,1\}\}$. Then, it is impossible for $\succ$ to be locally strict. This is because the set $\{0,1\}$ is both an indifference class, and it is open.}

\indent Proposition \ref{prop:proposition3} shows that for finite $X$, $\mathcal{P}^{ci}$ is empty, and for infinite $X$, $\mathcal{P}^{ci}$ is not Hausdorff. Therefore, restricting attention to locally strict preferences is not enough to recover the Hausdorff property, at least not for the topology we imposed on $X$.

\begin{Proposition}\label{prop:proposition3}
If $X$ is finite, then $\mathcal{P}^{ci}=\emptyset$. If $X$ is infinite, then $\mathcal{P}^{ci}$ is not Hausdorff in the final topology.
\end{Proposition}

\indent Example \ref{example:intuition} illustrates the intuition behind the proof:

\begin{example}\label{example:intuition}
Consider $X=\mathbb{R}$ and the sequence of utility functions, $u_n$, representing the ``more is better" preference, $\succ$.
\[
u_n(x)=
 \begin{cases} 
      \frac{1}{n}x + (1-\frac{1}{n})10 & 10 \geq x \\
      x & x>10
   \end{cases}.
\]
Notice that $\succ$ is locally strict since no subset of an indifference class is open in $\mathcal{T}_X(\succ)$. This sequence of utilities (and, therefore, the corresponding preference) converge to a preference, $\hat{\succ}$,  represented by the following utility function:
\[
u_n(x)=
 \begin{cases} 
      10 & 10 \geq x \\
      x & x>10
   \end{cases}.
\]
Notice that $\hat{\succ}$ is locally strict because no subset of an indifference class is open in $\mathcal{T}_X(\hat{\succ})$ but $\hat{\succ}\neq \succ$. Therefore, the constant sequence $(\succ)_{n\in\mathbb{N}}$ converges to $\hat{\succ}$, showing that the Hausdorff property fails.
\end{example}

\indent Combining Theorem \ref{th:theorem3} and Proposition \ref{prop:proposition3} suggests that recovering the Hausdorff property for non-strict preferences requires endowing $X$ with a preference-independent topology, and using this exogenous topology to define local strictness. Indeed, consider preference $\hat{\succ}$ in Example \ref{example:intuition}. Suppose that, instead of using $\mathcal{T}_X(\hat{\succ})$ as the topology for $\mathbb{R}$, we use the Euclidean topology for $\mathbb{R}$. Then, $\hat{\succ}$ is no longer locally strict. Indeed, $1\hat{\sim} 2$, and $N_1=(.75,1.25)$, $N_2=(1.75, 2.25)$ are neighborhoods of $1$ and $2$ respectively, such that $x\sim y$ for each $x\in N_1$ and $y\in N_2$. This suggests that the topology chosen for $X$ drives the negative result in Proposition \ref{prop:proposition3}. Consequently, we now ask whether the Hausdorff property may be recovered for the space of locally strict preferences if we endow $X$ with a fixed, preference-independent topology $\mathcal{T}_X$.

\subsection{Endowing $X$ with a preference-independent topology, $\mathcal{T}_X$}\label{sec:topology on X}

\indent This section examines if a fixed, preference-independent topology on $X$ ensures the Hausdorff property for the set of continuous and locally strict preferences, $\mathcal{P}^{cls}$. Example \ref{example:gap} demonstrates $\mathcal{P}^{cls}$'s failure to be Hausdorff. Theorem \ref{th:theorem4} proves an additional condition is necessary and sufficient for the Hausdorff property.

\indent To make the results in this section comparable those in Proposition \ref{prop:proposition3}, we limit our discussion to locally strict preferences. Additionally, we must ensure that preferences have a utility representation. Following Debreu \cite{debreu1954representation} and Rader \cite{rader1963existence}, we assume $(X, \mathcal{T}_X)$ is second countable topological space, and focus on \emph{continuous} preferences--those with open upper and lower contour sets (see definition \ref{def:continuous preference} below). 

\begin{Definition}\label{def:continuous preference}
Let $(X,\mathcal{T}_X)$ be a topological space, and let $\succ$ be a preference on $X$. We say $\succ$ is \emph{continuous} if for all $x\in X$ the following are true:
\begin{itemize}
\item[1.] $Lower (x)=\{y:x\succ y\}$ is open in $\mathcal{T}_X$.
\item[2.] $Upper(x)=\{y:y\succ x\}$ is open in $\mathcal{T}_X$.
\end{itemize}
\end{Definition}

\indent We first note that a preference $\succ\in\mathcal{P}^{cls}$ might have a \emph{gap}. We say a pair of alternatives $x,y\in X$ is a \emph{gap} if $x\succ y$ and $\{w:x\succ w\succ y\}=\emptyset$. Example \ref{example:gap} shows that $\mathcal{P}^{cls}$ is not Hausdorff because preferences may have gaps. %In Section \ref{sec:subjectivetopology} we do not consider the possibility of gaps because a continuous and locally strict preference cannot have a gap when the topology on $X$ is $\mathcal{T}(\succ)$.\footnote{See the proof of Proposition \ref{prop:proposition3} for details.} However, when $\mathcal{T}_X$ is preference-independent, gaps may arise even when preferences are continuous and locally strict.

\begin{example}\label{example:gap}
Let $X=[0,1]\cup[0,2]$ endowed with the subspace topology. That is, a set $G$ is open in $X$ if it can be written as $G=U\cap X$ where $U$ is open in $\mathbb{R}$.
Consider the following two preferences:
\begin{itemize}
\item $\succ$ is the ``more is better" preference, defined as $x\succ y\:\Leftrightarrow \: x>y$,
\item $\hat{\succ}$ is the ``more is better" preference with the exception that $1\:\hat{\sim}\:2$. Formally, $\hat{\succ}=\succ\setminus\{(1,2)\}$.
\end{itemize}
Clearly both preferences are continuous and locally strict. Furthermore, because $x\:\hat{\succ}\:y\Rightarrow x\succ y$, then $\succ\in N$ for any neighborhood $N$ of $\hat{\succ}$. Therefore, $\mathcal{P}^{cls}$ is not Hausdorff.
\end{example}

\indent Example \ref{example:gap} shows that excluding preferences with gaps is a necessary condition to recover the Hausdorff property. Theorem \ref{th:theorem4} below shows that it is also sufficient. 

\begin{Theorem}\label{th:theorem4}
Let $\mathcal{P}^{clsg}$ be the set of continuous, locally strict preference that have no gaps. $\mathcal{P}^{clsg}$ is Hausdorff in the final topology.
\end{Theorem}

\indent Together, Example \ref{example:gap} and Theorem \ref{th:theorem4} show that $\mathcal{P}^{clsg}$ is the largest Hausdorff subset of $\mathcal{P}^{cls}$. Building on Example \ref{example:intuition}, Example \ref{example:intuition2} illustrates the role that a preference-independent topology on $X$ plays in recovering the Hausdorff property.

\begin{example}\label{example:intuition2}
Consider $X=\mathbb{R}$ \emph{endowed with the Euclidean topology}, and the sequence of utility functions, $u_n$, representing the ``more is better" preference, $\succ$.
\[
u_n(x)=
 \begin{cases} 
      \frac{1}{n}x + (1-\frac{1}{n})10 & 10 \geq x \\
      x & x>10
   \end{cases}.
\]
Notice that $\succ$ is locally strict in the Euclidean topology. This sequence of utilities (and, therefore, the corresponding preferences) converge to a preference, $\hat{\succ}$,  represented by the following utility function:
\[
u_n(x)=
 \begin{cases} 
      10 & 10 \geq x \\
      x & x>10 
   \end{cases}.
\]

\indent Contrary to Example \ref{example:intuition}, $\hat{\succ}$ is \emph{not} locally strict: in the Euclidean topology, neighborhoods of $x=1$ and $y=2$ exist where every point in the neighborhood of $1$ is indifferent to every point in the neighborhood of $2$. Because the topology on $X$ is not given by the upper and lower contour sets of $\hat{\succ}$, but by the Euclidean topology, this example's notion of a ``neighborhood of $(x,y)=(1,2)$" differs from Example \ref{example:intuition}'s notion of neighborhood. Whereas $\hat{\succ}$ is locally strict in the space $(\mathbb{R},\mathcal{T}_\mathbb{R}(\hat{\succ}))$, it is not locally strict in the space $(\mathbb{R}, \Vert\cdot\Vert)$, where $\Vert\cdot\Vert$ denotes the Euclidean topology for $\mathbb{R}$. Therefore, $\hat{\succ}$ is not an admissible limit point of the constant sequence $(\succ)_{n\in\mathbb{N}}$.  
\end{example}

\indent Contrary to Theorem \ref{th:theorem3} and Proposition \ref{prop:proposition1}, Theorem \ref{th:theorem4} shows the space of continuous, locally strict preferences with no gaps is Hausdorff, despite $\mathcal{P}^{clsg}\supset\mathcal{P}^s$. Together, Theorem \ref{th:theorem3} and Proposition \ref{prop:proposition1} link the topological properties of preference spaces with economic assumptions on indifference curves: preference spaces are Hausdorff only if they exclude preferences with indifferences. Theorem \ref{th:theorem4} extends this analysis, connecting the Hausdorff property to both the economic assumptions on indifference curves \emph{and} the topological properties of the alternatives space. Unlike Theorem \ref{th:theorem3}, Theorem \ref{th:theorem4} only requires preferences to be strict in a \emph{local} sense rather than a global one. Because the definition of ``local" varies with the topology imposed on $X$, the topology chosen for $X$ plays a crucial role on the topological properties of the preference space.

\section{Conclusions}\label{sec:conclusions}

\indent In this paper we look at final topologies for preference spaces. Defined by a universal isomorphism, final topologies ensure that, across any model and its outputs, the continuous mappings from preferences to outputs and from utilities to outputs are isomorphic. Final topologies are important because the analysis of a model's continuity properties typically occurs in the space of utility functions, not the preference space itself. 

\indent Our main result, Theorem \ref{th:theorem2}, characterizes the basis of the final topology. Theorem \ref{th:theorem3} and Proposition \ref{prop:proposition1} then study the topological properties of the preference space. A trade-off exists between methodological generality, intuitive topological properties of the preference space, and the structure of indifference curves. Having a Hausdorff preference space that also allows for non-strict preferences proves impossible: the set of strict preferences is the largest set that is both Hausdorff and allows for all strict preferences.  Similarly, we find that the Hausdorff property and path-connectedness--two desirable properties for the preference space to have--are mutually exclusive.

\indent To circumvent these trade-offs, we can impose an exogenous topology on $X$ and focus on continuous, \emph{locally} strict preferences with no gaps. This approach ensures the final topology on preference spaces is Hausdorff and accommodates non-strict preferences, as evidenced by Theorem \ref{th:theorem4}. Consequently, specific topological assumptions on the space of alternatives have a strong impact on the final topology on preferences. 

\indent In summary, this paper comprehensively studies the topological properties of a preference space as a function of the allowable indifferences and topological assumptions on the set of alternatives. Understanding these is essential for interpreting a model's continuity with respect to utilities as indicative of continuity in the underlying preferences.

\bibliographystyle{abbrv}
\bibliography{topologiesbib}

\begin{thebibliography}{10}

\bibitem{border1994dynamic}
K.~C. Border and U.~Segal.
\newblock Dynamic consistency implies approximately expected utility
  preferences.
\newblock {\em Journal of Economic Theory}, 63(2):170--188, 1994.

\bibitem{borkovsky2010user}
R.~N. Borkovsky, U.~Doraszelski, and Y.~Kryukov.
\newblock A user's guide to solving dynamic stochastic games using the homotopy
  method.
\newblock {\em Operations Research}, 58(4-part-2):1116--1132, 2010.

\bibitem{chambers2019recovering}
C.~P. Chambers, F.~Echenique, and N.~S. Lambert.
\newblock Recovering preferences from finite data.
\newblock {\em Econometrica}, 89(4):1633--1664, 2021.

\bibitem{debreu1967neighboring}
G.~Debreu.
\newblock {\em Neighboring economic agents}.
\newblock University of Calif., 1967.

\bibitem{debreu1954representation}
G.~Debreu et~al.
\newblock Representation of a preference ordering by a numerical function.
\newblock {\em Decision processes}, 3:159--165, 1954.

\bibitem{dubra2000full}
J.~Dubra and F.~Echenique.
\newblock A full characterization of representable preferences.
\newblock {\em Documento de Trabajo/FCS-DE; 12/00}, 2000.

\bibitem{eaves1999general}
B.~C. Eaves and K.~Schmedders.
\newblock General equilibrium models and homotopy methods.
\newblock {\em Journal of Economic Dynamics and Control}, 23(9-10):1249--1279,
  1999.

\bibitem{grodal1974note}
B.~Grodal.
\newblock A note on the space of preference relations.
\newblock {\em Journal of Mathematical Economics}, 1(3):279--294, 1974.

\bibitem{hildenbrand1970economies}
W.~Hildenbrand.
\newblock On economies with many agents.
\newblock {\em Journal of economic theory}, 2(2):161--188, 1970.

\bibitem{kannai1970continuity}
Y.~Kannai.
\newblock Continuity properties of the core of a market.
\newblock {\em Econometrica: Journal of the Econometric Society}, pages
  791--815, 1970.

\bibitem{mas1974continuous}
A.~Mas-Colell.
\newblock Continuous and smooth consumers: Approximation theorems.
\newblock {\em Journal of Economic Theory}, 8(3):305--336, 1974.

\bibitem{nishimura2023class}
H.~Nishimura and E.~A. Ok.
\newblock A class of dissimilarity semimetrics for preference relations.
\newblock {\em Mathematics of Operations Research}, 2023.

\bibitem{rader1963existence}
T.~Rader.
\newblock The existence of a utility function to represent preferences.
\newblock {\em The Review of Economic Studies}, 30(3):229--232, 1963.

\bibitem{shiomura1998hicksian}
T.~Shiomura.
\newblock On the hicksian laws of comparative statics for the hicksian case:
  the path-following approach using an alternative homotopy.
\newblock {\em Computational Economics}, 12(1):25--33, 1998.

\end{thebibliography}

\appendix
\section{Proofs}

We start with a well known result. Under the final topology, $U$ is open in $\mathcal{P}$ if and only if $F^{-1}(U)$ is open in $\mathcal{U}$. Consequently, $F$ is continuous. Proofs are included for completeness.

\begin{Lemma}
$U$ is open in $\mathcal{P}$ if and only if $F^{-1}(U)$ is open in $\mathcal{U}$. As a consequence, $F$ is continuous.
\end{Lemma}
\begin{proof}
$(\Leftarrow)$ Let $U\subset\mathcal{P}$ be open. Let $g:\mathcal{P}\rightarrow\mathcal{P}$ be the identity function, so $g$ is continuous by definition. By definition, continuity of $g$ implies continuity of $(F\circ g)$. Because $(F\circ g)=F$, then $F$ is continuous. Thus, $F^{-1}(U)$ is open. \\
$(\Rightarrow)$ Let $U\subset\mathcal{P}$ and assume $F^{-1}(U)$ is open. Let $Z=\{0,1\}$ and let $\tau_Z=\{Z,\emptyset, \{1\}\}$ be the open sets in $Z$. Define $g:\mathcal{P}\rightarrow Z$ as $g(\succ)=1$ if and only if $\succ\in U$. Then $F\circ g$ is continuous, so that $g$ must be continuous. Therefore, $U=g^{-1}(1)$ is open.
\end{proof}

We now start characterizing the open sets in $\mathcal{P}$. We introduce two pieces of notation. First, for any preference, $\hat{\succ}$, and any set, $S\subset X$, we let $\hat{\succ}\vert S$ denote the restriction of $\hat{\succ}$ to $S$. Formally, $\hat{\succ}\vert S=\{(x,y)\in S\times S: x\hat{\succ} y\}$. Second, let $\succ\in\mathcal{P}$ and let $A\subset X$ such that for each pair $x,y\in A$ either $x\succ y$ or $y\succ x$. Then, we let $B(\succ,A)=\{\hat{\succ}: \hat{\succ}\vert A = \succ\vert A\}$.

\begin{Lemma}\label{lemma:Bopen}
Let $\succ\in\mathcal{P}$ and let $x,y\in X$ be such that $x\succ y$. Then, the set $B(\succ,\{x,y\})$ is open in $\mathcal{P}$.
\end{Lemma}
\begin{proof}
Pick any $\succ\in\mathcal{P}$ and $x,y\in A$ as in the statement of the lemma.
It suffices to show that $F^{-1}(B(\succ,\{x,y\}))$ is open in $\mathcal{U}$.
Notice that $F^{-1}(B(\succ,\{x,y\}))=\cup_{r\in\mathbb{R}} A(r)$ where $A(r)=\{u\in\mathcal{U}: u(x)>r>u(y)\}$.
Because $A(r)$ is open in $\mathcal{U}$ for each $r\in\mathbb{R}$, then $\cup_{r\in\mathbb{R}} A(r)$ is open in $\mathcal{U}$.
Indeed, $A(r)\equiv \Pi_{z\in X}U_z$ where $U_z=\mathbb{R}$ when $z\notin\{x,y\}$, $U_x=(r,\infty)$ and $U_y=(-\infty, r)$, which is open in the topology of pointwise convergence.
Therefore, $F^{-1}{B(\succ,\{x,y\})}$ is open in $\mathcal{U}$ and so $B(\succ,\{x,y\})$ is open in $\mathcal{P}$. 
\end{proof}

We now use Lemma \ref{lemma:Bopen} to show that $F$ is an open map. To do this, we show $F$ maps the basis of $\mathcal{U}$ into open sets in $\mathcal{P}$.

\begin{Lemma}\label{lemma:Fopen}
Let $U$ be an element in the basis of $\mathcal{U}$. Then, $F(U)$ is open in $\mathcal{P}$.
\end{Lemma}
\begin{proof}
Let $U$ be an element in the basis of $\mathcal{U}$.
Then, there exists a finite set $A\subset X$ such that $U=\Pi_{x\in X}I(x)$ where, for each $x\in X$, $I(x)\subset \mathbb{R}$ is open, $I(x)=\mathbb{R}$ if $x\notin A$, and $I(x)\neq\mathbb{R}$ if $x\in A$.
Furthermore, without loss of generality, $I(x)$ is an interval for each $x\in A$.
Therefore, for each $x\in A$ there are numbers $a_x,b_x\in\mathbb{R}$ such that $I_x=(a_x,b_x)$.
\\ \textbf{Case 1: $\cap_{x\in A} I(x)\neq\emptyset$.}
We will show that $F(U)=\mathcal{P}$, which is open and hence completes the proof.
Take any $\succ\in \mathcal{P}$. 
Let $(a,b)\subset \mathbb{R}$ be such that $(a,b)\subset \cap_{x\in A} I(x)$.
Such an interval exists because $\cap_{x\in A} I(x)$ is non empty and open.
Let $u\in F^{-1}(\succ)$. 
Then there exists a scaling $f:\mathbb{R}\rightarrow\mathbb{R}$ such that $f(u)(x)\in (a,b)$ for all $x\in A$ and $F(f(u))=F(u)=\succ$.
Because $f(u)\in U$ by construction and $F(f(u))=\succ$ then $\succ\in F(U)$.
Because $\succ\in\mathcal{P}$ was arbitrarily chosen this shows $\mathcal{P}\subset F(U)$, and this completes the proof.
\\ \textbf{Case 2: $\cap_{x\in A} I(x)=\emptyset$.}
Because $\cap_{x\in A} I(x)=\emptyset$ there exists $x_0,y_0\in A$ such that $I(x_0)\cap I(y_0)=\emptyset$.
Let $\mathcal{A}=\{\{x,y\}: I(x)\cap I(y)=\emptyset\}$.
Because $\{x_0,y_0\}\in\mathcal{A}$ then $\mathcal{A}\neq\emptyset$.
Thus, $(\forall u,v \in U)$ $(\forall \{x,y\}\in\mathcal{A})$, $u(x)\neq u(y)$, $v(x)\neq v(y)$ and $u(x) > u(y)\Leftrightarrow v(x)>v(y)$.
This follows from $u(x)\in I(x)$, $u(y)\in I(y)$ and $I(x)\cap I(y)=\emptyset$.
Let $\triangleright\in\mathcal{P}$ be any preferences such that, for each $\{x,y\}\in\mathcal{A}$, $x\triangleright y$ if and only if $I(x)>I(y)$ in the strong set order.\footnote{ To constructively build such a preference, for each pair $\{x,y\}\in\mathcal{A}$ we say that $x\triangleright y$ iff $I(x)>I(y)$ in the strong set order. By construction this relation is transitive: if $x\triangleright y$ and $y\triangleright z$ then (i) $\{x,y\}\in\mathcal{A}$, (ii) $\{y,z\}\in\mathcal{A}$ and (iii) $I(x) > I(y) >I(z)$. This is a transitive ranking of a finite set of pairs, so it can be extended to all of $X\times X$ as follows: (i) if $x,y\in A$ but $\{x,y\}\notin\mathcal{A}$ then neither $x\triangleright y$ nor $y\triangleright x$ (ii) if $x\notin A$ and $y\in A$, $y\triangleright x$ and (iii) if $x,y\notin A$, then neither $x\triangleright y$ nor $y\triangleright x$.} 
In particular, for all $\{x,y\}\in\mathcal{A}$, $x\triangleright y$ iff $u(x)>u(y)$ for all $u\in F(U)$.
Let $B(\triangleright,\mathcal{A})=\cap_{\{x,y\}\in\mathcal{A}}B(\triangleright,\{x,y\})$.
$B(\triangleright,\mathcal{A})$ is open because each $B(\triangleright,\{x,y\})$ is open and the intersection is finite.
We now show that $F(U)=B(\triangleright,\mathcal{A})$ which proves that $F(U)$ is open.\\ 
\underline{\emph{$F(U)\subset B(\triangleright,\mathcal{A})$:}} 
Let $\succ\in F(U)$.
Then there is $u\in U$ such that $F(u)=\succ$.
Because $u\in U$, then $F(u)\vert\{x,y\} = \triangleright\vert\{x,y\}$ for each $\{x,y\}\in\mathcal{A}$.
Thus, $\succ\vert\{x,y\} = \triangleright\vert\{x,y\}$ and thus $\succ\in B(\triangleright,\mathcal{A})$.
\\ \underline{\emph{$\succ\in B(\triangleright,\mathcal{A})\subset F(U)$:}} 
Pick $\succ\in B(\triangleright,\mathcal{A})$.
Thus $\succ\vert\{x,y\}=\triangleright\vert\{x,y\}$ for each pair $\{x,y\}\in\mathcal{A}$.
Thus, for any $\{x,y\}\in\mathcal{A}$ and any $u\in F^{-1}(\succ)$, $u(x)>u(y) \Leftrightarrow x\triangleright y$. 
Up to a rescaling, there exists $v\in F^{-1}(\succ)$ such that $v(x)\in I(x)$ for each $x\in A$.
Thus, there exists $v\in U\cap F^{-1}(\succ)$ so that $\succ=F(v)\in F(U)$.
\end{proof}

As a consequence of the previous Lemma, $F$ is an open map.\footnote{If $U$ is open in $\mathcal{U}$ then $U=\cup_{i\in I}D_i $ for some collection of basis elements $(D_i)_{i\in I}$. Then, $F(U)=\cup F(D_i)$. Since each $F(D_i)$ is open by assumption, then $F(U)$ is open, so $F$ maps open sets to open sets.} Consequently, if $\mathcal{D}$ is a basis for the topology of $\mathcal{U}$ then $F(\mathcal{D})$ is a basis for the final topology on $\mathcal{P}$. The previous lemma effectively shows $F(\mathcal{D})=\{B(\succ,\mathcal{A}): \succ\in\mathcal{P}\:,\:\mathcal{A}\in\alpha(\succ)\}$. This last observation proves Theorem \ref{th:theorem2}. For completeness, we include an explicit proof of Theorem \ref{th:theorem2} below.

Part of Lemma \ref{lemma:Fopen}'s proof involves showing that although sets of the form $B(\succ, A)$ generate the basis of $\mathcal{T}_P$, they do not form a basis. Example \ref{example:mathcalA} below shows this in a simple example.

\begin{example}\label{example:mathcalA} 
Suppose $X=\{1,2,3\}$.
The following set is an element of the basis of $\mathbb{R}^X$: $U=I_1\times I_2\times I_3$ where $I_1=(0,1)$, $I_2=(2,5)$, $I_3(2,5)$.
Pick any $u\in U$. 
Then $u$ represents a preference such that $2 \: F(u)\: 1$ and $3 \: F(u)\: 1$.
However, we could have $3 \:F(u) \:2$,  $2 \:F(u) \:3$, or neither.
Thus, pushing forward $U$ according to $F$ recover all preferences $\succ$ such that $2\succ 1$ and $3\succ 1$ but we cannot conclude anything about the ranking of $2$ and $3$.
Hence $F(U)=B(F(u), \{1,2\})\cap B(F(u),\{1,3\})$ where $u\in U$.
However, $B(F(u),\{1,2,3\})=B(F(u), \{1,2\})\cap B(F(u),\{1,3\})\cap B(F(u),\{2,3\}\neq B(F(u), \{1,2\})\cap B(F(u),\{1,3\})=F(U)$ because $U$ could include functions $v$ such that $u$ and $v$ disagree on how to rank elements $2$ and $3$.
For instance $u(1)=0.5$, $u(2)=3$, $u(3)=4$, and $v(1)=0.5$, $v(2)=4$, $v(3)=3$.

\end{example}

\textbf{Proof of Theorem \ref{th:theorem2}}
\begin{proof}
Let $\mathcal{D}$ be the basis for $\mathcal{U}$ generated by products of open intervals. 
That is, each $D\in\mathcal{D}$ is of the form $D=\Pi_x I_x$ where each $I_x$ is an open interval and $I_x\neq \mathbb{R}$ for a finite subset of elements $x\in X$.
We show $\mathcal{B}=F(\mathcal{D})$.
Because $F$ is open and $\mathcal{D}$ is a basis for $\mathcal{U}$, then $\mathcal{B}$ is a basis for $\mathcal{P}$.\newline
\emph{\underline{$\mathcal{B}\subset F(\mathcal{D})$:}}
Let $B(\succ,\mathcal{A})\in\mathcal{B}$, where $\succ\in\mathcal{P}$ and $\mathcal{A}\in\alpha(\succ)$.
Let $u$ be any representation of $\succ$.
Let $A\in\mathcal{A}$, so that $A$ is $\succ$-strictly ranked.
Let $\varepsilon=\min\{\vert u(x)-u(y)\vert: x,y\in A, \: x\neq y\}>0$, where positivity is guaranteed because $A$ is finite and $\succ$-strictly ranked.
For each $x\in A$, let $I_x=(u(x)-\frac{\varepsilon}{2}, u(x)+\frac{\varepsilon}{2})\subset \mathbb{R}$.
For each $x\notin A$, let $I_x= \mathbb{R}$.
Notice that $I_x\cap I_y=\emptyset$ for all $x,y\in A$, $x\neq y$.
Let $D=\Pi_x I_x \in \mathcal{D}$.
Then, $B(\succ, A)=F (D)\in F(\mathcal{D})$ for each $A\in\mathcal{A}$.
Therefore $B(\succ, \mathcal{A})=\cap_{A\in\mathcal{A}}B(\succ,A)\in F(\mathcal{D})$ so $\mathcal{B}\subset F(\mathcal{D})$. \footnote{This follow because products and finite intersections commute. Therefore, the finite intersection of base elements in the product topology is itself a base element.}
\newline

\emph{\underline{$\mathcal{B}\supset F(\mathcal{D})$:}}
Pick any $F(D)$ where $D\in \mathcal{D}$.
Let $A=\{x: I_x\neq \mathbb{R}\}$, which is finite by definition.
First, assume $\cap_{x\in A} I_x\neq \emptyset$.
Then, $\cap_{x\in A} I_x=(a,b)$ for some $a,b\in\mathbb{R}$.
Thus, $F(D)=\mathcal{P}\in \mathcal{B}$.
Second, assume $\cap_{x\in A} I_x= \emptyset$.
Then there exists $A_0\subset A$ such that $\cap_{x\in A_0}I_x=\emptyset$.
Let $\mathcal{A}=\{A_0\subset A: \cap_{x\in A_0}I_x=\emptyset\}$
Therefore, if $v,u\in D$ and $x,y\in A_0$ then $u(x)\neq u(y)$, $v(x)\neq v(y)$ and $u(x)>u(y) \Leftrightarrow v(x)>v(y)$.
That is, any utility function in $D$ ranks the alternatives in each $A_0\in\mathcal{A}$ the same way and no two alternatives receive equal utility.
Then $F(D)=\cap_{A_0\in\mathcal{A}}B(\succ,A_0)$ for all $\succ\in F(D)$.
Therefore, $F(D)\in\mathcal{B}$ for all $D\in\mathcal{D}$ and so $F(\mathcal{D})\subset\mathcal{B}$.
\end{proof}

\subsection{Proof of Theorem \ref{th:theorem3} and Proposition \ref{prop:proposition1}}

\begin{theoremthree}
The space $\mathcal{P}^s$ is Hausdorff and path disconnected.
\end{theoremthree}
\begin{proof}
\indent We first show $\mathcal{P}^s$ is Hausdorff.
Let $\succ_0,\succ_1\in\mathcal{P}^s$ such that $\succ_1\neq\succ_0$.
Then there exists $x,y\in X$ such that $x\succ_0 y$ and $y\succ_1 x$.
Consider the sets $B_0=B(\succ_0,\{x,y\})$ and $B_1=B(\succ_1,\{x,y\})$.
By Theorem \ref{th:theorem2}, $B_0$ and $B_1$ are neighbordoods of $\succ_0$ and $\succ_1$ respectively.
Furthermore, $B_0\cap B_1=\emptyset$. 
 This shows $\mathcal{P}^s$ is Hausdorff.

\indent We now argue that $\mathcal{P}^s$ is totally path disconnected.
It suffices to show that $\mathcal{U}^s$ is totally path-disconnected.
Take $u,v\in\mathcal{U}^s$ to be distinct functions.
Assume there is a continuous path $t$ joining $u$ and $v$.
That is, $t:[0,1]\rightarrow \mathcal{U}^s$ with $t(0)=u$, $t(1)=v$ and $t$ continuous.
Since $u\neq v$ and $u,v\in\mathcal{U}^s$, there are points $x,y\in X$ such that $u(x)>u(y)$ and $v(y)>v(x)$.
Let $\Delta (s)=t(s)(x)-t(s)(y)$ for $s\in [0,1]$.
Then, $\Delta(0)>0>\Delta(1)$.
Thus, there is a point $s^*$ such that $\Delta(s^*)=0$.
Then, $t(s^*)(x)=t(s^*)(y)$, contradicting that $t(s^*)\in\mathcal{U}^s$.
Thus, no two distinct functions are connected via a continuous path.

\end{proof}

\begin{propone}
Let $\mathcal{P}^0$ be Hausdorff and $\mathcal{P}^s\subset\mathcal{P}^0$. Then, $\mathcal{P}^0=\mathcal{P}^s$.
\end{propone}
\begin{proof}
We proceed in two steps.
Recall that $\simeq$ denotes the total indifference preference. 
\\ \emph{\underline{Step 1: if $\simeq\in\mathcal{P}^0$ then $\mathcal{P}^{0}$ is not Hausdorff.}}
By Theorem \ref{th:theorem2}, the only open neighborhood of $\simeq$ is $\mathcal{P}^0$.
Therefore, $\simeq$ is topologically indistinguishable from any other preference, and thus $\mathcal{P}^0$ is not Hausdorff.\\
We may alternatively proceed using limit arguments.
Let $\textbf{0}$ denote constant function $0$. 
Clearly $F(\textbf{0})=\simeq$.
Let $\succ\in\mathcal{P}^s$ be any preference and let $u$ be any representation of $\succ$.
Then $u_n=\frac{1}{n}u$ represents $\succ$.
Since $u_n\rightarrow \textbf{0}$, $F(u_n)=\succ$ for all $n$, $F(\textbf{0})=\simeq$, then $\simeq\in\bar{\{\succ\}}$,so $\mathcal{P}^0$ cannot be Hausdorff.\\
\\ \emph{\underline{Step 2: if $\simeq\notin\mathcal{P}^0$ then $P^{0}$ is not Hausdorff.}}
Take $\succ\in\mathcal{P}^0\setminus\mathcal{P}^s$.
For each $x\in X$, let $<x>$ be the indifference class of $x$ according to $\succ$.
For each $<x>$, let $f_{<x>}:<x>\rightarrow\mathbb{R}$ be any function such that $f(a)\neq f(b)$ for all $a,b\in<x>$.\footnote{Such $f$ exists because $<x>\subset X$ and $\mathcal{P}^s\neq\emptyset$.}
Construct $\succ_0\in \mathcal{P}^s$ as follows: for any pair $x,y\in X$, if $x\succ y$ then $x\succ_0 y$; if $x\sim y$ then $x\succ_0y$ if and only if $f_{<x>}(x)>f_{<x>}(y)$.
Notice that $\succ_0\in\mathcal{P}^s$ and that $\succ_0\vert A=\succ\vert A$ for any finite set $A$ that is $\succ$-strictly ranked.
Let $G$ be any open set such that $\succ\in G$.
Then there is a finite collection $\mathcal{A}$ and a preference $\succ_1\in\mathcal{P}_0$ such that $\succ\in B(\succ_1, \mathcal{A})$. 
Without loss of generality $B(\succ_1, \mathcal{A})=B(\succ, \mathcal{A})$.\footnote{Indeed, $\succ \in B(\succ_1, A)$ for all $A\in \mathcal{A}$ implies $\succ\vert A=\succ_1\vert A$ for all $A\in\mathcal{A}$ and so $B(\succ_1, A)=B(\succ, A)$ for all $A\in \mathcal{A}$.}
By construction, $\succ_0\in B(\succ, \mathcal{A})$. 
Thus $\succ$ and $\succ_0$ are topologically indistinguishable, so $\mathcal{P}_0$ is not Hausdorff.
\end{proof}

\subsection{Proof of Proposition \ref{prop:proposition3} and Theorem \ref{th:theorem4}}

For our next lemma, we start by recalling notation we introduced in the main text. For all $x\in X$, $<x>$ represents the indifference class of $x$. Then, define $\mathcal{P}^{ci}=\{\succ: \:(\forall x\in X)\: if \: A\subset<x>\:\Rightarrow\:\text{$A$ is not open}  \}$.

\begin{Lemma}\label{lemma:locally strict}
Let $\succ\in\mathcal{P}^{ci}$. Then, $\succ$ is locally strict in $(X,\mathcal{T}_X(\succ))$.
\end{Lemma}
\begin{proof}
Let $(x,y)\in X\times X$ be such that $x\succsim y$.
If $x\succ y$, then the locally strict property holds vacuously.
Assume then that $x\sim y$.
Then, let $V$ be a neighborhood of $(x,y)$.
Then, there is an open set $G=G_1\times G_2\subset X\times X$ such that $(x,y)\in G\subset V$.
Assume that for all $(z,z')\in G_1\times G_2$ we had $z\sim z'$.
Because $x\in G_1$, for any $z'\in G_2$ we get $x\sim z'$.
Thus, $G_2\subset <x>$.
By analogous reasoning, $G_1\subset <y>$.
Then, $G_1\subset <y>$ and $G_2\subset <x>$, contradicting that $G$ is open.
Hence, there must exist $(z,z')\in G$ such that $z\succ z'$, completing the proof.
\end{proof}

\begin{proptwo}
If $X$ is finite, then $\mathcal{P}^{ci}=\emptyset$. If $X$ is infinite, then $\mathcal{P}^{ci}$ is not Hausdorff in the final topology.
\end{proptwo}
\begin{proof}
\textbf{Part 1: If $X$ is finite, then $\mathcal{P}^{ci}=\emptyset$.}\\
Let $\succ$ be a preference that is locally strict.
We start with two observations about $\succ$: first, $\succ\neq\simeq$ and, second $\succ$ must contain an infinite number of distinct indifference classes.
Indeed, assume that $\succ$ only contains finitely many distinct indifference classes (say, $N$ indifference classes).
Then we can pick a representative from each indifference class and label them in an increasing way: $x_1\succ x_2 \succ (...) \succ x_N$.
Then $<x_1>$ is open as it is the upper contour set of $x_2$, thus $\succ$ is not locally strict.
We conclude that if $\succ$ is locally strict, then it has infinitely many (distinct) indifference classes.\footnote{In fact, it must have uncountably many distinct indifference classes.}
Consequently, if $\succ$ is locally strict, then $X$ must have infinitely many members.
Alternatively, if $X$ is finite, then $\mathcal{P}^{ci}=\emptyset$.\\
\textbf{Part 2: If $X$ is infinite, then $\mathcal{P}^{ci}$ is not Hausdorff in the final topology.}
We start with an observation about the topology $\mathcal{T}_X(\succ)$.
The basis for such a topology is given by sets of the form $I(z,y)=\{x:y\succ x\succ z\}$, where $y,z\in X\cup\{\infty,-\infty\}$, and we use the convention $I(-\infty,y)\equiv Lower(y)$ and $I(z,\infty)\equiv Upper(z)$.\footnote{This convention is for notational simplicity only, so as not to have different notation for intervals and upper/lower contour sets.}
Furthermore, we say a point $x$ is \emph{isolated} if the following is true: there are points $a$ and $b$ in $X$ such that (1) $a\succ x\succ b$, (2) $I(x,a)=I(b,x)=\emptyset.$
Next, we make two observations about locally strict preferences.
First, a locally strict preference cannot have isolated points.
If $x$ was an isolated point then $<x>=I(a,b)$ for the points $a$ and $b$ that make $x$ isolated.
This is a contradiction because it would imply $<x>$ is open.
Second, if $\succ$ is locally strict, $\succ$ must have infinitely many distinct indifference classes.
This is a consequence of the arguments in Part $1$.
Consequently, there is a point $x\in X$ such that $x$ is neither maximal nor minimal (\emph{i.e.}, there exist $a,b\in X$ such that $a\succ x\succ b$) because otherwise every point would be either maximal or minimal and there would only be two indifference classes, contradicting that $\succ$ is locally strict.
Putting these observations together we conclude that any locally strict preference must have a point that is neither maximal nor minimal, and this point cannot be isolated.
We now use these observations to prove that $\mathcal{P}^{ci}$ is not Hausdorff.
Pick any triplet $x,y,z\in X$ and choose any preference $\succ\in\mathcal{P}^{ci}$ that satisfies $z\succ x\succ y$.
Such a preference exists because every preference must have an $x$ that is neither maximal nor minimal.
Because $x$ cannot be isolated then at least one of the following must hold: either $I(x,z')\neq\emptyset$ for all $z'\in Upper(x)$, or $I(y',x)\neq\emptyset$ for all $y'\in Lower(x)$.
We will now construct a preference $\succ_1$ such that $\succ_1\neq \succ$ and $\succ$ belongs to every neighborhood of $\succ_1$, thus showing $\mathcal{P}^{ci}$ is not Haursdorff.
Without loss of generality, assume $I(x,z')\neq\emptyset$ for all $z'\in Upper(x)$; we can carry out an analogous construction for the other case.
For each pair $a,b\in X$ define $\succ_1$ as follows: if $a\in Upper(x)$ and $b\in Upper(x)$ then $\succ\vert\{a,b\}=\succ_1\vert\{a,b\}$, $a\in Upper(x)$ and $b\notin Upper(x)$ then $a\succ_1 b$, and if $a\notin Upper(x)$, $b\notin Upper(x)$ then $a\sim_1 b$.
We make three observations about $\succ_1$. 
First, $\succ_1\neq\succ$ because $x\sim_1 y$ and $x\succ y$.
Second, $\succ_1$ is locally strict.
To see this it suffices to show that no subset of the indifference class of $x$ under $\succ_1$ is open in $\mathcal{T}_X(\succ_1)$.
Let $G\subset \{w\in X:w\sim_1 x\}$ be an open set.
Because $G$ is open there must be a non-empty set $I(a,b)$ such that $I(a,b)\subset G$.\footnote{This is because sets of the form $I(a,b)$ with $b\succ_1 a$ form a basis for $\mathcal{T}_X(\succ_1)$.}
Because every point in $G$ is minimal by construction, $a=-\infty$.
Then, there must be an alternative $b\in X$ such that $I(-\infty,b)=Lower(b)\subset G$.
Because $I(x,z')\neq\emptyset$ for all $z'\in Upper(x)$, such a $b\in X$ cannot exist: if $b\in \{w:w\sim_1 x\}$ then $Lower(b)=\emptyset$, and if $b\succ_1 x$ then there exists $b'\in Lower(b)\cap Upper(x)$ so $Lower(b)\nsubset G$.
Thus, $G$ cannot be open.
Lastly, $\succ_1$ satisfies the following property: $a \succ_1 b \Rightarrow a\succ b$ for all $a,b\in X$.
This last observation implies that if $N$ is a neighborhood of $\succ_1$ then $\succ\in N$.
Therefore, $\mathcal{P}^{ci}$ is not Hausdorff.

\end{proof}

\begin{theoremfour}\nonumber
The space $\mathcal{P}^{clsg}$ is Hausdorff.
\end{theoremfour}
\begin{proof}
Let $\succ,\hat{\succ}\in\mathcal{P}^{clsg}$, $\succ\neq\hat{\succ}$.
We will show that there are neighborhoods of $\succ$ and $\hat{\succ}$ that are disjoint.
Because $\succ\neq\hat{\succ}$, there must be $x,y\in X$ such that $x\succ y$ and $x\:\hat{\nsucc}\:y$.
\\
\underline{Case 1: $y\:\hat{\succ}\:x$.} Then $B(\succ,\{x,y\})$ and $B(\hat{\succ},\{x,y\})$ are disjoint neighborhoods of $\succ$ and $\hat{\succ}$ respectively, and this concludes the proof.\\
\underline{Case 2: $y\:\hat{\sim}\:x$.} 
Because $\succ$ has no gaps, then $\{w:x\succ w\succ y\}\neq\emptyset$.
Furthermore, $\{w:x\succ w\succ y\}$ is open in $\mathcal{T}_X$ by continuity of $\succ$.
First, observe that there must be $z\in \{w:x\succ w\succ y\}$ such that $z\:\hat{\nsim}\:x$, otherwise $\{w:x\succ w\succ y\}$ would be an open set such that $z\:\hat{\sim}\:z'$ for all $z,z'\in \{w:x\succ w\succ y\}$, contradicting that $\hat{\succ}$ is locally strict.
Take $z\in \{w:x\succ w\succ y\}$ such that $z\:\hat{\nsim}\: x$.
If $x\:\hat{\succ}\:z$ then we get $y\:\hat{\sim}\:x\:\hat{\succ}\:z$. Thus, $y\:\hat{\succ}\:z$ and $z\succ y$, implying $B(\succ,\{z,y\})$ and $B(\hat{\succ},\{z,y\})$ are disjoint neighborhoods of $\succ$ and $\hat{\succ}$ thereby concluding the proof.
If $z\:\hat{\succ}\:x$ then we get $z\:\hat{\succ}\:x$ and $x\succ z$, implying  $B(\succ,\{z,x\})$ and $B(\hat{\succ},\{z,x\})$ are disjoint neighborhoods of $\succ$ and $\hat{\succ}$ thereby concluding the proof.

Therefore, we can always neighborhoods of $\succ$ and $\hat{\succ}$ that are disjoint, and this concludes the proof. 

\end{proof}

\end{document}